\documentclass[12pt, final]{l4dc2022}

\title{Bounding the difference between model predictive control and neural networks}
\usepackage{times}
\usepackage{amsmath,bm}
\usepackage{bbding}
\usepackage{mathtools}
\usepackage{amsmath,amsfonts,mathtools}
\usepackage{pifont}
\usepackage{wasysym}
\usepackage{amssymb}
\usepackage{times}
\usepackage{amsmath}
\usepackage{mathtools}

\usepackage{amssymb}
\usepackage{graphicx}
\usepackage{graphics}
\usepackage{caption}
\usepackage{setspace}
\usepackage{bytefield}
\newtheorem{problem}{Problem}

\def\PP{\textcolor{green!50!blue}}
\DeclareMathOperator{\blkdiag}{blkdiag}

\makeatletter
 \let\Ginclude@graphics\@org@Ginclude@graphics 
\makeatother

\author{ \Name{Ross Drummond} \Email{ross.drummond@eng.ox.ac.uk}\\
 \Name{Stephen R. Duncan} \Email{stephen.duncan@eng.ox.ac.uk}\\
 \addr Department of Engineering Science, University of Oxford, 17 Parks Road, OX1 3PJ, Oxford, United Kingdom \\
 \AND
 \Name{Matthew C. Turner} \Email{m.c.turner@soton.ac.uk}\\
 \addr Department of Electronics and Computer Science,
   University of Southampton, Southampton, UK, SO17 1BJ
    \AND
 \Name{Patricia Pauli} \Email{patricia.pauli@ist.uni-stuttgart.de}\\
 \Name{Frank Allg\"ower} \Email{frank.allgower@ist.uni-stuttgart.de}\\
 \addr Institute for Systems Theory and Automatic Control, University of Stuttgart, 70569 Stuttgart, Germany
   }

\begin{document}

\maketitle

\begin{abstract}

There is a growing debate on whether the future of feedback control systems will be dominated by data-driven or model-driven approaches. Each of these two approaches has their own complimentary set of advantages and disadvantages, however, only limited attempts have, so far, been developed to bridge the gap between them. To address this issue, this paper introduces a method to bound the worst-case error between feedback control policies based upon model predictive control (MPC) and neural networks (NNs). This result is leveraged into an approach to automatically synthesize MPC policies minimising the worst-case error with respect to a NN. Numerical examples highlight the application of the bounds, with the goal of the paper being to encourage a more quantitative understanding of the relationship between  data-driven and model-driven control. 

\end{abstract}

\begin{keywords}%
  Neural network robustness, model predictive control, semi-definite programming.
\end{keywords}

\section{Introduction}

Two distinct approaches for the feedback control of dynamical systems are emerging, categorised into \emph{model-driven} and \emph{data-driven} methods. Model-driven approaches rely upon a physical model of the system to make predictions about its future behaviour and then compute a control input to shape this future response (\cite{green2012linear}). Data-driven methods, in particular model-free reinforcement learning, forgo the need of a predictive model and instead rely upon identifying patterns within the trajectory data itself to determine control policies (\cite{sutton2018reinforcement}). The benefits and limitations of these two approaches are also quite distinct. Model-driven approaches are often equipped with performance certificates for properties such as closed-loop stability, optimality and robustness (\cite{green2012linear}), with these certificates making them more  suited for safety-critical systems. But model-driven methods generally struggle in applications where existing models remain underdeveloped and when  applied to large-scale systems as computing the control polices can become intractable (\cite{vsiljak2005control}).

Data-driven methods, by contrast, do not require an accurate model of the system  (\cite{sutton2018reinforcement}) and can find control policies which would be unobtainable using model-based approaches (as they avoid the biases inherent within models), a point illustrated by the famous ``move 37'' from AlphaGo (\cite{silver2017mastering}) which confounded experts' analysis. A main drawback of most data-driven methods is their general lack of rigorous guarantees, e.g. robustness \citep{szegedy2013intriguing}, as small changes in their inputs can lead to large changes in their outputs, and explaining their behaviour remains an open problem. Robustness issues are a primary reason why we have yet to fully see the heralded successes of data-driven methods transfer from controlled, laboratory settings to real-life practical applications. 

The distinction between model and data driven control methods has driven interest in bridging the gap between them. In particular,  there is a need  to better characterise the relationship between each approach's benchmark methods which are, arguably,  model predictive control (MPC) for model-based control and deep neural networks (NNs) for the data-driven approach. Relating these two approaches will not only bring greater understanding but will also encourage the development of new control methods that combine each approach's complementary advantages.

\paragraph{Contribution:} The two main results of this paper are:
\begin{enumerate}
    \item A method to bound the difference between a given MPC and a NN controller (Theorem \ref{thm:approx}).
    \item A method to synthesise MPCs that robustly approximate neural networks (Theorem \ref{thm:syn}).
\end{enumerate}
The overriding theme of these results is to quantify the similarity between MPC  and neural network control actions, with the bounds holding  robustly, in the sense that they hold for fairly generic neural network architectures and MPC cost functions considering pre-defined hyper-rectangles as input constraints. The bounds are obtained by solving a semi-definite program (SDP), a class of convex optimisation programs  which have been extensively used for analysing the robustness of both NNs and MPC, as in \cite{fazlyab2020safety, li2006stability}.

\paragraph{Literature: } Several existing results have also been concerned with relating the disparate control schemes of MPC and NNs. Of particular note are the results on  explicit MPC where the user takes advantage of the fact that the optimal solution to the linear MPC problem can be expressed as a piece-wise linear function (\cite{bemporad2002explicit}). By identifying which region of this piece-wise linear function the current state of the system is in, known as solving the ``point location problem'', then the MPC action can be obtained without online computation of the finite horizon optimal control problem.  The main drawback of explicit MPC  is that solving the point location problem can be challenging (since the complexity of the piece-wise affine function, measured by the number of piece-wise regions, scales exponentially with the state dimension and the number of constraints, \cite{alessio2009survey}), which has limited its use in practice. To avoid this computational bottleneck, it has been proposed to approximate the piece-wise linear optimal MPC  policy with a neural network, \cite{parisini1995receding,zhang2019safe}, as the approximating NN will also be a piece-wise linear function if a ReLU function is used for the activation function \PP{\citep{darup2020exact}}. Evaluating a NN is generally  much simpler than solving the point location problem but can introduce an error which has proven difficult to bound. The bounds formulated in this paper are able to quantify the worst-case error of this approximation.

 The presented results can be understood within the more general context of applying  control theory to verify the robustness of neural networks. The theoretical framework for these methods are now classic, emerging from studies such as \cite{chu}, \cite{barabanov}, \cite{angeli2009convergence}, and have been extended in recent years in studies such as \cite{fazlyab2020safety}. The presented results contrast with these existing robustness problems in that they are concerned with the robustness of a neural network with respect to MPC, not to itself.  Some of the authors used similar ideas to quantify the similarity of one neural network with respect to another (\cite{li2021robust}), which allowed the approximation errors of pruned, quantised and reduced-order neural networks to be bounded (\cite{li2021robust,drummond2021reduced}).  The recently released study \cite{fabiani2021reliably}  which was also concerned with bounding the error between MPC and a NN must also be highlighted. The method developed in that paper involved solving a mixed integer optimisation problem to bound the error between MPC and NNs with ReLU activation functions in particular. In contrast, the results of this paper solve a SDP to verify the robustness of fairly generic NN architectures with respect to MPC. Finally, it is highlighted how this paper is concerned with quantifying the similarity between MPC and a NN but it is known, that, in certain circumstances, these two policies can be shown to be equivalent, \cite{darup2020exact}.

Examples of NN robustness failures (\cite{szegedy2013intriguing}) have driven mounting  interest in using semi-definite programming for NN  analysis and synthesis \cite{fazlyab2020safety}. These results include the estimation of Lipschitz constants for NN mappings \citep{fazlyab2019efficient, pauli2021training} and the stability analysis of  closed-loop systems using NN controllers \cite{yin2021stability}. The stability analysis  problem in particular can be linked to the problem of absolute stability (\cite{desoer2009feedback, barabanov, chu}), as a feed-forward neural network is a static, memoryless nonlinearity \cite{pauli2021linear}. Within this context, \cite{yin2021imitation} proposed approximating a MPC using NN controllers while satisfying LMI constraints to ensure closed-loop stability. Enforcing stability on a large-scale structure like a NN can be challenging, which motivates a workaround that finds and enforces bounds relating the two controller types. Doing so may lead to more elegant stability constraints to be formulated, as meeting these bounds have already been shown to imply stability guarantees \cite{hertneck2018learning}.

\subsection*{Notation}
The set of real numbers is denoted $\mathbb{R}$ and the natural numbers are $\mathbb{N}$. For dimension $n$, the set of real vectors is   $\mathbb{R}^n$, the set of vectors with non-negative elements is denoted $ \mathbb{R}^n_+$ and the zero vector is  ${\bm{0}_{n}}$. Element-wise positivity of a vector is denoted $> $ while positive definiteness of a matrix is denoted $\succ$, with similar definitions for negativity and semi-definiteness of matrices. The set of real matrices of dimension $n \times m$ is  $ \mathbb{R}^{n \times m}$ and the zero matrix is $\bm{0}_{n \times m}$. For dimension $n$, the space of positive (semi-)definite matrices is  ($\mathcal{S}^n_{\succeq 0}$) $\mathcal{S}^n_{\succ 0}$,  non-negative diagonal matrices is $\mathbb{D}^{n}_{\succeq 0}$ and the identity matrix is $I_n$. We use the $\star$ notation for symmetric matrices, as in $\begin{bmatrix}
 A & B \\ B^\top & C 
\end{bmatrix}= \begin{bmatrix}
 A^\top & B^\top \\ B & C^\top 
\end{bmatrix}= \begin{bmatrix}
 A & B \\ \star & C 
\end{bmatrix}.$
Several of the proofs of the lemmas are immediate and so are not shown. 

\section{Problem Set-up}

This section introduces the general problem considered in the paper of computing bounds and approximations between MPCs and NNs.

\subsection{Model predictive control}

Consider input-constrained model predictive control with  linear dynamics and a quadratic cost.

\begin{definition}\label{Def:MPC-1}
Consider a system with linear dynamics
\begin{align}\label{lin_dyns}
    x[k+1] = Ax[k] + Bu[k],
\end{align}
with state $x [k]\in \mathbb{R}^{n_x}$ and input $u[k]\in \mathbb{R}$. For some finite horizon $N \in \mathbb{N}$, the input constrained MPC control law with linear dynamics is defined as the solution to
{\small \begin{subequations}\label{MPC_opt}\begin{align}
z_{\text{MPC}}(x[k]) =  & \,\text{arg}\,\min_u \begin{bmatrix}
x[k+1] \\ \vdots \\ x[k+N]
\end{bmatrix}^\top Q
 \begin{bmatrix} x[k+1] \\ \vdots \\ x[k+N] \end{bmatrix}
+ 
\begin{bmatrix}
u[k] \\ \vdots \\u[k+N-1] 
\end{bmatrix}^\top R \begin{bmatrix} u[k] \\ \vdots \\ u[k+N-1]
\end{bmatrix},
\\
  & \text{subject to }Lu \leq b,~ b \geq 0,
\end{align}\end{subequations}}
for some $Q \in \mathbb{S}^{N n_x}_{\succeq 0} $, $R \in \mathbb{S}^{N}_{\succeq 0} $,  $L \in \mathbb{R}^{n_{in}\times N}$ and $b \in \mathbb{R}^{n_{in}}_+$.
\end{definition}

The linear dynamics of \eqref{lin_dyns} mean that  the optimisation problem \eqref{MPC_opt} can be converted into a quadratic program (QP). 
\begin{definition}\label{def:MPC-QP}
The MPC control law defined by Definition \ref{Def:MPC-1} can be expressed as the solution to a constrained quadratic program
\begin{subequations}\label{MPC_P}\begin{align}
    z_\text{MPC}(x[k])  = & ~ \text{arg}\,\min_{u} u^\top Hu  + u^\top h x[k] 
    \\
  & \text{subject to }Lu \leq b,~ b \geq 0 ,\label{ineq_MPC}
\end{align}
\end{subequations}
for some $H \in \mathcal{S}^{N}_{\succ 0} $, $h \in \mathbb{R}^{N \times n_x} $ parameterised by the matrices ($A,\,B,\,Q,\,R$) of Definition \ref{Def:MPC-1} and $b \in \mathbb{R}^{n_{in}}_+$. The first instant of the QP solution from \eqref{MPC_P} is taken as the MPC control action, as in 
\begin{align}
    u_\text{MPC}(x[k])= \begin{bmatrix} 1, & {\bm{0}_{N-1}}^\top\end{bmatrix}z_\text{MPC}(x[k]) = C_\text{MPC} z_\text{MPC}(x[k]). 
\end{align}
\end{definition}


\subsection{Neural networks}

The class of neural networks considered in this paper are those which can be structured in an  implicit form (\cite{el2021implicit}).
\begin{definition}\label{def:NN}
The considered class of neural networks $u_\text{NN}(x[k]): \mathbb{R}^{n_x} \to \mathbb{R}$ are those defined by the implicit form
\begin{subequations}\label{NN}
\begin{align}
  {z}_\text{NN} & =   \phi (W {z}_\text{NN} + W_0 x[k] + \beta ), \quad \phi(\cdot): \mathbb{R}^M \mapsto \mathbb{R}^M, \\
   u_\text{NN}(x[k])         & =   C_\text{NN} {z}_\text{NN} + D_\text{NN} ,
 \end{align}
where $ C_\text{NN} = \begin{bmatrix} 
       \bm{0}_{n_\ell}^\top,  & \hdots\,, &   \bm{0}_{n_\ell}^\top, & W^{o} 
       \end{bmatrix},$
       $D_\text{NN} \in \mathbb{R}$ and 
{\small \begin{align}
{z}_\text{NN} & = \begin{bmatrix}
   z^1_\text{NN} \\ z^2_\text{NN} \\ \vdots \\ z^\ell_\text{NN}
  \end{bmatrix}, \,
 W =  \begin{bmatrix}
    W^{1,1}   & W^{1,2} & \hdots & W^{1,\ell}  \\
     W^{2,1} & W^{2,2} & \ddots & 
     \vdots\\
     \vdots & \ddots & \ddots  & W^{\ell-1,\ell} \\
     W^{\ell,1}  & \dots & W^{\ell,\ell-1} & W^{\ell,\ell} 
    \end{bmatrix} ,\,
    W_0  = \begin{bmatrix}
              W^{1,0} \\
              W^{2,0} \\
              \vdots 
              \\ W^{\ell,0}
             \end{bmatrix} ,\,
 \beta  = \begin{bmatrix}
   \beta^0 \\ \beta^1 \\ \vdots \\ \beta^{\ell}     \end{bmatrix}.
\end{align}}
\end{subequations}
\end{definition}
In the above, $\phi(\cdot)$ are the nonlinear activation functions which could be any function satisfying one or more of the function properties of Definition \ref{def:prop} (Appendix 1), acting element-wise upon their arguments. The NN will be assumed to have $\ell$ hidden layers and, for the sake of notational simplicity, it will be assumed that each layer will be of dimension $n_\ell$, as in $z^j_\text{NN} \in \mathbb{R}^{n_\ell}$, $W^{i,j}\in \mathbb{R}^{n_\ell \times n_\ell}$,  $W^{i,0}\in \mathbb{R}^{n_\ell \times n_x}$, $\beta^j \in \mathbb{R}^{n_\ell}$ and ${W^o}^\top\in \mathbb{R}^{n_\ell}$ for $i = 1,\,\dots, \, \ell$,  $j = 1,\,\dots, \, \ell$. The total length of the vector $z_\text{NN} \in \mathbb{R}^M$ is defined as $M = n_\ell \times \ell$.

Adopting the implicit NN structure of \eqref{NN} in Definition \ref{def:NN}, also known as an equilibrium network  \cite{revay2020lipschitz,bai2019deep}, helps simplify the notation whilst also ensuring  the results are applicable to a wide class of NN architectures, including deep recurrent and feed-forward NNs.

\subsection{Bounds between  MPC and NN}

Bounds  quantifying the similarity between the NN and MPC control policies described in Definitions \ref{def:MPC-QP} and \ref{def:NN} are sought.  These bounds are formalised as solutions to the following problem. 
\begin{problem}\label{prob:approx}
For any state $x[k] \in \mathcal{X}$, find $\gamma_x \geq 0$ and $\gamma \geq 0$  such that the error between  the neural network and MPC control actions, namely $u_\text{NN}(x[k])-u_\text{MPC}(x[k])$, is bounded by
\begin{align}\label{bnd}
\|u_\text{NN}(x[k])-u_\text{MPC}(x[k])\|_2^2 \leq \gamma_x\|x[k]\|_2^2+ \gamma,\quad \forall x[k] \in \mathcal{X}.
\end{align}
\end{problem}

\begin{remark}\label{rem:gamma}
If both controllers enforce $x[k] = 0$ as an equilibrium point of the dynamical system, then it is possible to obtain $\gamma = 0$ since the two control actions should be equivalent at that point. 
\end{remark}


\section{Quadratic Constraints}

For the bounds of Problem \ref{prob:approx} to hold robustly, as in for all inputs $x[k]$ in some set $\mathcal{X}$, characterisations of both the nonlinear activation functions of the neural network and  the solution structure of the MPC control law have to be included within the problem formulation. Here, this information is incorporated  using quadratic constraints. 

\subsection{Quadratic constraints for the neural network}

The first step of this process is to determine the quadratic constraint of the nonlinear activation functions $\phi(.)$. Many different activation functions have been implemented within neural networks, including the tanh and ReLU functions, with each having their own characteristics. In general, each of these candidate activation functions satisfies certain properties which can be incorporated within a robust optimisation framework to compute the error bounds. 

Examples of these function properties are given in Definition \ref{def:prop} (Appendix 1). Most commonly adopted activation functions, including the tanh and ReLU, satisfy at least one of these properties. For example, the tanh function is slope-restricted, bounded and sector bounded, the ReLU function is slope-restricted while both it and its complement are positive and satisfy the complementarity conditions.  


By satisfying one or more of these function properties, the activation functions also satisfy quadratic constraints, as detailed in Lemma \ref{lem:prop} (Appendix 2). These quadratic constraints are defined with the multipliers ${\bf T}^i$ ($i \in \left\{\text{s},\text{sl},+,c+,B,0,\times,\otimes \right\}$), and in the following, the notation   ${\bf T} \in \mathbb{T}$ will be used to collect all of the relevant ${\bf T}^i$'s ($i \in \left\{\text{s},\text{sl},+,c+,B,0,\times,\otimes \right\}$) from Lemma \ref{lem:prop} (Appendix 2), with the set $\mathbb{T}$ characterising the positivity conditions of the ${\bf T}^i$'s.

If the neural network's activation function satisfies more than one of these quadratic constraints, as in it satisfies one or more of the properties in Lemma \ref{lem:prop} (Appendix 2), then each of  these quadratic constraints can be combined into a single quadratic constraint defined by a matrix $\lambda$, as considered previously in results such as \cite{drummond2021reduced,fazlyab2020safety}.  In this way, information about the nonlinear activation functions of the neural network can be incorporated into the robust optimisation problem for the error bounds.

\begin{lemma}\label{lem:ass_phi}
If the activation function $\phi(\cdot)$ satisfies one or more of the quadratic constraints of Lemma \ref{lem:prop} (Appendix 2), then there exists matrices 
\begin{align}\label{lambda_hat}
\hat{\lambda}({\bf T}) = \begin{bmatrix} \hat{\lambda}_{11}&    \hat{\lambda}_{12} &  \hat{\lambda}_{13} 
\\
\star &   \hat{\lambda}_{22} &  \hat{\lambda}_{23}  
\\
\star &  \star &  \hat{\lambda}_{33}
\end{bmatrix}    ,
\quad 
{\lambda}({\bf T}) = \begin{bmatrix} {\lambda}_{11}&    {\lambda}_{12} &  {\lambda}_{13} 
\\
\star &   {\lambda}_{22} &  {\lambda}_{23}  
\\
\star &  \star &  {\lambda}_{33}
\end{bmatrix} ,
\end{align}
defined by the ${\bf T}^i$'s $(i \in \left\{\text{s},\text{sl},+,c+,B,0,\times,\otimes \right\})$ of Lemma \ref{lem:prop} (Appendix 2), such that with the vectors
\begin{subequations}\begin{align}\label{ineq:lem}
    \hat{\mu}(x[k])  & =  \begin{bmatrix}
      (W z_\text{NN} + W_0 x[k] + \beta)^\top, & {z_\text{NN}}^\top, & 1
    \end{bmatrix}^\top,
    \\
         {\mu}(x[k])  & = \begin{bmatrix}
{x[k]}^\top,  & {z_\text{NN}}^\top, & 1
\end{bmatrix}^\top,
\end{align}\end{subequations}
  the following inequality holds
\begin{align}
s_\text{NN}(x[k]) = \hat{\mu}(x[k])^T \hat{\lambda}({\bf T}) \hat{\mu}(x[k])
=\mu(x[k])^T \lambda({\bf T}) \mu(x[k])
 \geq 0.
\end{align} 

\end{lemma}
\begin{proof}
{Immediate from applying the S-procedure (\cite{jonsson2001lecture}) on  Lemma \ref{lem:prop} (Appendix 2).}
\end{proof}

\subsection{Quadratic constraints for model predictive control}

Following \cite[Result 5]{li2006stability}, a quadratic constraint can also be characterised for the solution of an MPC control law defined by the QP of Definition \ref{def:MPC-QP}. In \cite{li2006stability}, this quadratic constraint was introduced to characterise the closed-loop stability of linear systems controlled by MPC, using Zames-Falb multipliers to provide the stability certificates (\cite{zames1968stability,turner2021discrete}). Here, we exploit this quadratic constraint  to include information about the solution structure of the MPC control law into the error bound computation.  

\begin{theorem}\label{thm:Will}[\cite{li2006stability}, Result 5]
The MPC control law characterised by the QP of Definition~\ref{def:MPC-QP} satisfies, for any $\tau_\text{QP} \geq 0$,
    \begin{align}\label{s_QP_1}
s_\text{QP}(x[k])  =-\tau_\text{QP}\left({z_{\text{MPC}}}^\top Hz_{\text{MPC}} + {z_{\text{MPC}}}^\top hx[k]\right) \geq 0.
\end{align}
\end{theorem}


On top of this quadratic constraint for the QP cost function \eqref{s_QP_1}, another quadratic constraint for the inequality bounds \eqref{ineq_MPC} can also be devised since, for any positive vector $q\in \mathbb{R}^{n_{in}}_+$,
\begin{align}\label{eq:s_in}
    s_{in}(x[k]) = q^\top  \left(b- Lz_{MPC}\right) +  \left(b- Lz_{MPC}\right)^\top q \geq 0.
\end{align}
The function properties of Definition \ref{def:prop} (Appendix 1) can be used to build satisfactory $q$. For example, if $\phi(.)$ and its complement are positive, as for the ReLU, then, for some $\tau_{in}\in \mathbb{R}_+^{n_{in}}$, $(\tau_{in}^{+},\,\tau_{in}^{c+})\in \mathbb{R}_+^{ n_{in}\times M}$, one could define
\begin{align}
    q = {\tau_{in}} 
  +(\tau_{in}^{+}+\tau_{in}^{c+}(I_M-W)) z_\text{NN} - \tau_{in}^{c+}W_0 x[k] - \tau_{in}^{c+}\beta
\end{align}
which links the NN and MPC vectors, $z_\text{NN}$ and $z_\text{MPC}$ within the quadratic constraint.
The quadratic constraint for the inequality constraints $ s_{in}(x[k])$ can then be combined with that of the cost function from Theorem \ref{thm:Will} to give the following characterisation of the MPC solution. 
\begin{lemma}\label{lem:MPC_qc}
For any $x[k] \in \mathbb{R}^{n_x}$, the MPC solution of the QP given in Definition \ref{def:MPC-QP} satisfies
  \begin{align} \label{s_QP}
s_\text{MPC}(x[k]) =~ &  s_\text{QP}(x[k])+ s_\text{in}(x[k]) \geq 0,
\end{align}
with $s_\text{QP}(x[k])$ defined in \eqref{s_QP_1} and $s_\text{in}(x[k])$ in \eqref{eq:s_in}.
\end{lemma}

\subsection{Quadratic constraints for the input}

Information about the input constraint set $x[k] \in \mathcal{X}$ must also be included within the robust optimisation problem for the bounds. By restricting the input space to a hyper-rectangle, the following lemma defines a quadratic constraint for this set containment.

\begin{definition}
For some upper $\overline{x}$ and lower $\underline{x}$ bound, define $\mathcal{X}:= \{x[k]:  \, x[k] \geq \underline{x},\,x[k] \leq \overline{x} \}$.
\end{definition}
\begin{lemma}[\cite{fazlyab2020safety}]\label{lem:x_bound}
If $x[k]\in \mathcal{X}$, then, for any $\tau_x \in \mathbb{D}_+^{n_x}$, the following holds
\begin{align}
    s_{\mathcal{X}}(x[k]) = \begin{bmatrix} x[k]^\top & 1 \end{bmatrix}
    \begin{bmatrix} -2\tau_x & \tau_x (\overline{x} + \underline{x})  \\ \star & -(\overline{x}^\top\tau_x\underline{x}+\underline{x}^\top\tau_x\overline{x} ) \end{bmatrix}
    \begin{bmatrix} x[k] \\ 1 \end{bmatrix} \geq 0.
\end{align}
\end{lemma}

\section{Error bound between MPC and NN}

The quadratic constraints of Lemmas \ref{lem:ass_phi}, \ref{lem:MPC_qc} and \ref{lem:x_bound} allow information about the nonlinear activation functions, the solution structure of the MPC controller and the set containment $x[k] \in \mathcal{X}$ to be included within the error bound optimisation. In this way, the error bounds of Problem \ref{prob:approx} can be computed by solving the following SDP.

\begin{theorem}\label{thm:approx}
Consider a given NN structure  of Definition \ref{def:NN} and MPC control law of Definition \ref{def:MPC-QP}. For  given weights $\omega_x$ and  $\omega$, if there exists a solution to 
\begin{subequations}\begin{align}
\min_{{\bf T},\,\tau_{QP},\,\tau_{in},\, \gamma_x, \,\gamma, \,\tau_x} &\quad \omega_x\gamma_x + \omega\gamma
\\
\text{subject to: } &\quad  \Lambda_\text{NN}({\bf T}) + \Lambda_\text{MPC}(\tau_{QP},\tau_{in})+ X(\tau_x)  + \Omega(\gamma_x,\,\gamma) \preceq 0,\label{LMI_11}
\\
 &\quad  {\bf T} \in \mathbb{T}, \tau_x \in \mathbb{D}^{n_x}_+,\,\tau_{QP} \geq 0, \tau_{in} \in \mathbb{R}_{+}^{n_{in}}, \, \gamma_x \geq 0 , \gamma \geq 0,
\end{align}
\end{subequations}
where the matrices $\Lambda_\text{NN}({\bf T})$, $\Lambda_\text{MPC}(\tau_{QP},\tau_{in})$, $ X(\tau_x)$ and  $\Omega(\gamma_x,\,\gamma)$ are defined in \eqref{mats} of Appendix 3, then 
\begin{align} \label{thm1_bnd_orig}
    \|u_\text{NN}(x[k])-u_\text{MPC}(x[k])\|_2^2 \leq \gamma_x\|x[k]\|_2^2+ \gamma,\quad \forall x[k] \in \mathcal{X}.
\end{align}
\end{theorem}
\begin{proof}
Multiplying \eqref{LMI_11} on the left by 
\begin{align}
\zeta(x[k]) = \begin{bmatrix}
x[k]^\top,  & {z_\text{NN}}^\top,& {z_\text{MPC}}^\top,  & 1
\end{bmatrix}^\top
\end{align}
and on the right by it's transpose implies that, when $x[k] \in \mathcal{X}$, then
\begin{align}
    s_{NN}(x[k]) + s_{\mathcal{X}}(x[k])+ s_{MPC}(x[k])+  \|u_\text{NN}(x[k])-u_\text{MPC}(x[k])\|_2^2 -\gamma_x\|x[k]\|_2^2- \gamma \leq 0.
\end{align}
Lemmas \ref{lem:ass_phi}, \ref{lem:MPC_qc} and \ref{lem:x_bound} state that $s_{NN}(x[k]) \geq 0$, $s_{\mathcal{X}}(x[k]) \geq 0 $, $s_{MPC}(x[k]) \geq 0 $ when $x[k] \in \mathcal{X}$. Hence, the bound \eqref{thm1_bnd_orig} must hold. 
\end{proof}


\section{Automatic synthesis of an MPC approximating a NN}

Theorem \ref{thm:approx}  allows the error between \emph{given} MPC and NN control policies to be bounded. This idea can be extended to \emph{synthesize} MPC control laws that minimise the worst-case error to the NN. 

\begin{theorem}\label{thm:syn}
Consider a given NN structure  of \eqref{NN}. For the given weights $\omega_x$ and  $\omega$, if there exists a solution to 
\begin{subequations} \label{cost_syn}\begin{align}
\min_{{\bf T},\,H,\,h,\,L,\,b,\, \gamma_x, \,\gamma, \,\tau_x} &\quad \omega_x\gamma_x + \omega\gamma
\\
\text{subject to: } &\quad  \Lambda_\text{NN}({\bf T}) + \check{\Lambda}_\text{MPC}(H,h,L,b)+ X(\tau_x) + \Omega(\gamma_x,\,\gamma)  \preceq 0,\label{LMI_1}
\\
 &\quad  {\bf T} \in \mathbb{T}, \tau_x \in \mathbb{D}_+^{n_x}, \, \gamma_x \geq 0 , \gamma \geq 0, \, H \in \mathcal{S}^{N}_{\succ 0},  \, h \in \mathbb{R}^{N \times n_x},\,b \in \mathbb{D}_{+}^{n_{in}},  L \in \mathbb{R}^{{n_{in}}\times N}
\end{align}
\end{subequations}
with the matrices $\Lambda_\text{NN}({\bf T})$, $\check{\Lambda}_\text{MPC}(H,h,L,b)$, $ X(\tau_x)$ and  $\Omega(\gamma_x,\,\gamma)$ defined in Appendix 3, then 
\begin{align} \label{thm1_bnd}
    \|u_\text{NN}(x[k])-u_\text{MPC}(x[k])\|_2^2 \leq \gamma_x\|x[k]\|_2^2+ \gamma,\quad \forall x[k] \in \mathcal{X}.
\end{align}
\end{theorem}
\begin{proof}
From Lemma \ref{lem:MPC_qc}, the computed $(H,h,L,b)$ define a quadratic program from Definition \ref{def:MPC-QP} for which $s_{MPC}(x[k])\geq 0$ using $\tau_\text{QP} = 1$ and each element of $\tau_\text{in}$ being one. The proof is then immediate from the argument of Theorem \ref{thm:approx}.
\end{proof}

\begin{remark}
The main advantage of using Theorem \ref{thm:syn} to generate an approximating MPC is that it includes the worst-case approximation error  directly within the cost function being minimised in~\eqref{cost_syn}. 
\end{remark}

\begin{figure}
    \centering
    \includegraphics[width=\textwidth]{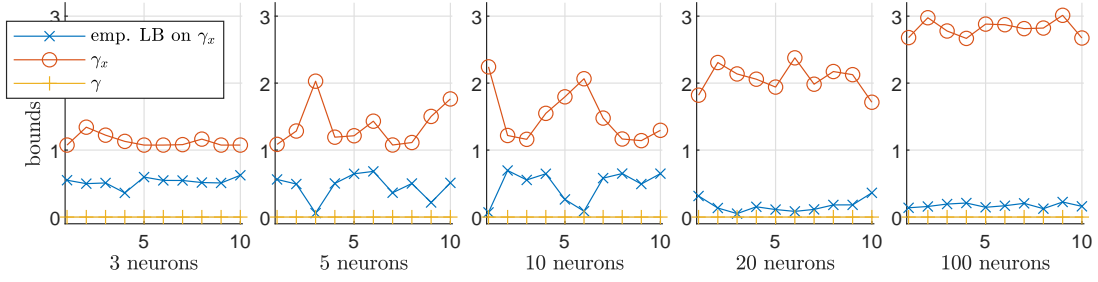}
    \caption{Bounds $\gamma$ and $\gamma_x$ and empirical lower bound on $\gamma_x$ for respectively 10 NNs with $\{3,5,10,20,100\}$ neurons in the hidden layer.}
    \label{fig:bounds}
\end{figure}

\section{Numerical examples}
In the following, we apply Theorem \ref{thm:approx} to bound the difference between an MPC and a NN controller. All NNs were trained using Pytorch and all SDPs were  solved in Matlab using the parser \texttt{YALMIP} (\cite{yalmip}) and solver \texttt{MOSEK} (\cite{mosek}).

We consider the following linear time-invariant system
\begin{equation}
    x[k+1]=
        \begin{bmatrix}
            4/3 & -2/3\\
            1 & 0
        \end{bmatrix}
    x[k]+
        \begin{bmatrix}
            0 \\ 1
        \end{bmatrix}
    u[k],
\end{equation}
and solve \eqref{MPC_P} choosing the horizon length $N=10$ and the input constraints $-0.1\leq u_\text{MPC}\leq~0.1$ such that $L=\begin{bmatrix}1 & -1 \end{bmatrix}^\top$, $b=\begin{bmatrix}0.1 & 0.1 \end{bmatrix}^\top$. The MPC quadratic cost function of \eqref{MPC_opt} was parameterised by
$$
\widetilde{P} =  \begin{bmatrix}7.1667  & -4.2222\\ -4.2222 & 4.6852\end{bmatrix},~    \widetilde{Q} = \begin{bmatrix}1  & -2/3\\ -2/3  & 3/2\end{bmatrix},~
\widetilde{R} = 1,
$$
yielding $Q=\blkdiag(\widetilde{Q},\,\widetilde{Q},\,\dots,\widetilde{Q},\,\widetilde{P})$ and $R=\blkdiag(\tilde{R},\,\dots,\,\tilde{R})$. 
Subsequently, feed-forward neural networks  with one hidden layer and ReLU activation were trained using back-propagation to approximate this MPC controller. The results are detailed in Figure \ref{fig:bounds} where upper and lower bounds for the obtained approximation error are shown. The upper bound, defined in terms of  $\gamma$ and $\gamma_x$, was obtained by solving Theorem \ref{thm:approx} while the empirical lower bound was obtained by sampling. For different numbers of neurons in the hidden layer, ten randomly initialised NNs were trained and their approximation errors are plotted in the figure. For all cases, it was found that $\gamma =0$, agreeing with the prediction of Remark \ref{rem:gamma}. The computed upper bound was always above the empirical lower bound, with the bounds being relatively tight for the small NNs. The larger gap between the upper and lower bounds for the larger NN examples is due to the increasing number of neurons causing an increased abstraction of the neural network mapping by the quadratic constraints, and hence an increase in the conservatism of the problem.

\section*{Conclusions}
A method to bound the worst-case error between feedback control policies based upon neural networks (NNs) and model-predictive control (MPC) was introduced. Using these bounds, a method to synthesize MPC policies minimising the worst-case error with respect to a NN was developed. The overriding goal of this work was to improve our understanding of, and even quantify, the relationship between MPC and NNs, helping to bridge the gap between model-driven and data-driven control. Future work will explore using the bounds to generate closed-loop stability certificates of both the MPC and NN controllers, reducing the conservatism of the bounds and applying them to the verification of control systems used in practice.

\section{Appendix}\label{app:mats}

\subsection{Appendix 1: Activation function properties}\label{app:def_phi}
\begin{definition}[\cite{drummond2021reduced}]\label{def:prop}
The activation function $\phi(s): \mathcal{S} \subset \mathbb{R} \to \mathbb{R}$ satisfying $\phi(0) = 0$ is said to be \emph{sector bounded} if 
$
\frac{\phi(s)}{s}~\in~[0, \delta] ~ \forall s \in \mathcal{S}, ~\delta > 0,
$
and \emph{slope restricted} if 
$
\frac{d\phi(s)}{ds}~\in~[\underline{\beta}, \beta],  ~\forall s \in \mathcal{S}, ~ \beta >0.
$
If $\underline{\beta} = 0$ then the nonlinearity is \emph{monotonic} and if $\phi(s)$ is slope restricted then it is also sector bounded. The activation function $\phi(s)$ is \emph{bounded} if 
$
\phi(s)~\in~[\underline{c}, \overline{c}],~ \forall s \in \mathcal{S},
$
it is \emph{positive} if
$
\phi(s) \geq 0 , ~ \forall s \in \mathcal{S},
$
its \emph{complement is positive} if
$
\phi(s)-s \geq 0 , ~ \forall s \in \mathcal{S},
$
and it satisfies the \emph{complementarity condition} if 
$
(\phi(s)-s)\phi(s) = 0, ~ \forall s \in \mathcal{S}.
$
\end{definition}

\subsection{Appendix 2: Quadratic constraints for the activation functions}\label{app:qcs}

\begin{lemma}[\cite{drummond2021reduced}]\label{lem:prop}
Consider the vectors $y,y_1 \in \mathbb{R}^{n_y}$,  and $\upsilon \in \mathbb{R}^{n_{\upsilon}}$ that are mapped component-wise through the activation functions $\phi(\cdot):\mathbb{R}^{n_y} \to \mathbb{R}^{n_y}$ and $\tilde{\phi}(\cdot):\mathbb{R}^{n_v} \to \mathbb{R}^{n_v}$. If $\phi(y)$ is
\emph{sector-bounded}, then
\begin{subequations}\label{qi_gen}\begin{align}
 (\delta y-\phi(y))^T{\bf T}^{\text{s}}\phi(y)  & \geq 0,\quad \forall y \in \mathbb{R}^{n_y}, ~{\bf T}^{\text{s}}  \in \mathbb{D}^{n_y}_+ ; \label{sec_quad}
\end{align}
\emph{slope-restricted} then
 \begin{equation}
(\beta (y-y_1)-(\phi(y)-\phi(y_1))^T{\bf T}^{\text{sl}} (\phi(y)-\phi(y_1)-\underline{\beta}(y-y_1))  \geq 0;
     \quad   \forall \{y, \, y_1 \} \in \mathbb{R}^{n_y},\,  {\bf T}^{\text{sl}}  \in \mathbb{D}^{n_y}_+;
 \end{equation}
\emph{bounded} then 
\begin{align}
 (\overline{c}-\phi(y))^T{\bf T}^{B}(\phi(y)-\underline{c})  & \geq 0,\quad \forall y \in \mathbb{R}^{n_y}, \, {\bf T}^{B}\in \mathbb{D}^{n_y}_+;
\end{align}
\emph{positive} then
\begin{align}
({\bf T}^{+})^T  \phi(y)\geq 0,  \quad \forall y \in \mathbb{R}^{n_y},\, {\bf T}^{+} \in \mathbb{R}_+^{n_y};
\end{align}
such that its \emph{complement is positive} then 
\begin{align}
({\bf T}^{c+})^T (\phi(y)-y) \geq 0,  \quad \forall y \in \mathbb{R}^{n_y}, \, {\bf T}^{c+}  \in \mathbb{R}^{n_y}_+.
\end{align}
If $\phi(y)$ satisfies the \emph{complementary} condition then 
\begin{align}
 (\phi(y)-y)^T{\bf T}^{0}\phi(y)  = 0,\quad \forall y \in \mathbb{R}^{n_y},\, {\bf T}^{0} \in \mathbb{D}^{n_y}.\label{comp_quad}
\end{align}
Additionally, if both $\phi(y)$ and $\tilde{\phi}(\upsilon)$ and their complements are positive then so are the \emph{cross terms} 
\begin{align}\label{cross_qcs_1}
\small
 \tilde{\phi}(\upsilon)^T{\bf T}^{\times}(\phi(y)-y) \geq 0,  \; \forall \upsilon \in \mathbb{R}^{n_v}, y \in \mathbb{R}^{n_y}, \, {\bf T}^{\times} \in \mathbb{R}_+^{n_{\upsilon} \times n_y}, \\
 \tilde{\phi}(\upsilon)^T{\bf T}^{\otimes}\phi(y) \geq 0,  \; \forall \upsilon \in \mathbb{R}^{n_v}, y \in \mathbb{R}^{n_y}, \, {\bf T}^{\otimes} \in \mathbb{R}_+^{n_{\upsilon} \times n_y}  .\label{cross_qcs_2}
\end{align}\end{subequations}
\end{lemma}

\subsection{Appendix 3:  Matrix definitions}\label{app:qcs}

{\small \begin{subequations}\label{mats}\begin{align}
     \Omega = 
        \begin{bmatrix}-\gamma_x I_{n_x}  & \bm{0}_{n_{x} \times M}  &\bm{0}_{n_{x} \times N} & \bm{0}_{n_{x} }
    \\ \star & {C_\text{NN}}^\top  C_\text{NN} & -{C_\text{NN}}^\top C_\text{MPC}  & \bm{0}_{M}
    \\ \star & \star & {C_\text{MPC}}^\top  C_\text{MPC}  &  -{C_\text{MPC}}^\top{D_\text{NN}} \\
    \star & \star & \star & -\gamma
    \end{bmatrix},
~
    \Lambda_{NN} = 
    \begin{bmatrix} {\lambda}_{11}  & {\lambda}_{12}  &\bm{0}_{n_{x} \times N} & {\lambda}_{13} 
    \\ \star & {\lambda}_{22} & \bm{0}_{M \times N}  & {\lambda}_{23}   
    \\ \star & \star & \bm{0}_{N \times N}  &  \bm{0}_{N} \\
    \star & \star & \star & {\lambda}_{33}
    \end{bmatrix},
 \end{align}
 \begin{align}\Lambda_{MPC} = 
    \begin{bmatrix} \bm{0}_{n_x \times n_x}  & \bm{0}_{n_x \times M}  & -(\tau_{QP}h^\top)/2 & \bm{0}_{n_x}  
    \\ \star & \bm{0}_{M \times M} & \bm{0}_{M \times N}  & \bm{0}_{M}   
    \\ \star & \star &  -\tau_{QP}H  &  -L^\top \tau_{in} \\
    \star & \star & \star & b^\top\tau_{in} + \tau_{in}^\top b
    \end{bmatrix},
     \end{align}
 \begin{align}
    X = 
    \begin{bmatrix} -2\tau_x & \bm{0}_{n_{x} \times M} & \bm{0}_{n_{x} \times N} & \tau_x (\overline{x} + \underline{x})   \\ \star & \bm{0}_{M \times M} & \bm{0}_{M \times N} & \bm{0}_{M}
    \\ \star & \star & \bm{0}_{N \times N} & \bm{0}_{N}
    \\  \star & \star & \star &  -(\overline{x}^\top\tau_x\underline{x}+\underline{x}^\top\tau_x\overline{x} ) \end{bmatrix},
    ~
   \check{\Lambda}_\text{MPC} =  
   \begin{bmatrix} \bm{0}_{n_x \times n_x}  & \bm{0}_{n_x \times M}  & -h^\top/2 & \bm{0}_{n_x}  
    \\ \star & \bm{0}_{M \times M} & \bm{0}_{M \times N}  & \bm{0}_{M}   
    \\ \star & \star &  -H  &  -L^\top \\
    \star & \star & \star & b + b^\top
    \end{bmatrix}.
\end{align}

\end{subequations}}

\acks{Ross Drummond would like to thank the Royal academy of engineering for funding this research through a UKIC Fellowship.
This work was partly funded by Deutsche Forschungsgemeinschaft (DFG, German Research Foundation) under Germany’s Excellence Strategy - EXC
2075 - 390740016. We acknowledge the support by the Stuttgart Center for Simulation Science (SimTech). The authors thank the International Max Planck Research School for Intelligent Systems (IMPRS-IS) for supporting Patricia Pauli.
}

\bibliography{bibliog} 

\end{document}